\documentclass[11pt]{article}
\usepackage[margin=1in]{geometry}
\usepackage{sty}

\newcommand{\texorpdfstring}[2]{#1}

\title{A semidefinite program for unbalanced multisection in the stochastic block model}

\usepackage{authblk}
\author[1]{Amelia Perry\thanks{Email: {\tt ameliaperry@mit.edu}.}}
\author[1]{Alexander S.\ Wein\thanks{Email: {\tt awein@mit.edu}. This research was conducted with Government support under and awarded by DoD, Air Force Office of Scientific Research, National Defense Science and Engineering Graduate (NDSEG) Fellowship, 32 CFR 168a.}}
\affil[1]{Massachusetts Institute of Technology, Department of Mathematics}

\begin{document}
\maketitle\thispagestyle{empty}

\begin{abstract}
We propose a semidefinite programming (SDP) algorithm for community detection in the stochastic block model, a popular model for networks with latent community structure. We prove that our algorithm achieves exact recovery of the latent communities, up to the information-theoretic limits determined by Abbe and Sandon \cite{as}. Our result extends prior SDP approaches by allowing for many communities of different sizes. By virtue of a semidefinite approach, our algorithms succeed against a semirandom variant of the stochastic block model, guaranteeing a form of robustness and generalization. We further explore how semirandom models can lend insight into both the strengths and limitations of SDPs in this setting.
\end{abstract}

\newpage\setcounter{page}{1}

\section{Introduction}

\subsection{Community detection}

The stochastic block model is among the most actively studied models for networks with latent community structure. Such models lend themselves to the statistical task of recovering the community structure given a single such graph, a task known as {\bf community detection}. This task manifests itself as that of finding human communities in social networks, and that of determining sets of proteins that function together in protein--protein interaction networks; more generally, community detection is the analogue of clustering for network data.

Since the introduction of the stochastic block model by Holland \cite{holland}, a wide range of algorithmic approaches have been deployed. Such algorithms are generally compared through the range of model parameters in which they succeed in recovering information; whether they achieve {\bf exact recovery} of the underlying community structure, or only {\bf partial recovery} of community structure correlated with the truth; and their computational efficiency.

In recent years, the block model has received a surge of interest as a meeting point of machine learning, algorithms, and statistical mechanics. A series of conjectures by Decelle et al.\ \cite{decelle} predicted a sharp threshold, separating a range of model parameters in which partial recovery is possible from a range in which it is impossible, using non-rigorous techniques from statistical physics. Some of these threshold effects have been proven \cite{mns,massoulie}, along with analogues for exact recovery \cite{abh,as}, and among the first algorithms to achieve exact recovery all the way up to these thresholds were those based on the powerful algorithmic tools of semidefinite programming (SDP). Indeed, the maximum likelihood estimation problems for block models amount to cut problems on graphs, which have a rich history of study through SDPs (for example \cite{gw,fj}); it is natural for these to appear among the most powerful tools for community detection.

In recent work of Abbe and Sandon \cite{as}, alternative methods of spectral clustering and local refinement have taken the lead in exact recovery, providing very general and sharp algorithmic results. It is natural to wonder how well SDP approaches can match these results, particularly in the light of strong robustness properties that convex programs enjoy. In this paper, we extend the frontier of optimal results on SDPs for exact recovery to the case of multiple communities of unequal sizes. We further address the relationship of SDPs to \emph{semirandom models}, highlighting their robustness as an advantage over spectral techniques, but also exposing how this same robustness can pose fundamental limits to SDP approaches.

\subsection{Models}

\subsubsection{Stochastic block model}

The {\bf stochastic block model} is a generative model for graphs, in which we suppose a vertex set of size $n$ has been partitioned into $r$ disjoint `communities' $S_1, \ldots, S_r$. An undirected graph is randomly generated from this hidden partition as follows: every unordered pair $(u,v)$ of distinct vertices is independently connected by an edge with probability $Q_{ij}$ where $u,v$ belong to communities $i,j$ respectively. Here $Q$ is a symmetric $r \times r$ matrix. Given a graph sampled from this model, the goal is to recover the underlying partition.

Many papers specialize to the {\bf planted partition model}, the case where $Q_{ii} = p$ for all $i$ and $Q_{ij} = q$ for all $i \ne j$. We will largely work with this specialization, and occasionally discuss the more general model.

\subsubsection{Semirandom models}
Much work in community detection focuses on the Erd\H os--R\'enyi-style random models defined above. These models are mathematically convenient, owing to the independence of edges, yet most real-world networks do not take this form. Moreover, there is cause for concern that some algorithms developed for the models above are highly brittle and do not generalize to other graph models: as illustrated in \cite{rjm}, one spectral method with strong theoretical guarantees degrades very rapidly when a sparse planted partition model is perturbed by adding small cliques, which are otherwise very infrequent in this random model, but do occur frequently in many real-world graphs.

Extending beyond the random models above, Blum and Spencer \cite{bs} introduced the notion of {\bf semirandom models} for graph problems, as an intermediate model between average-case and worst-case performance. The idea, extended by Feige and Kilian \cite{fk}, is to generate a preliminary graph according to a random model, but then allow a {\bf monotone adversary} to make unbounded, arbitrarily structured changes of a nature that should only help the algorithm by further revealing the ground truth. Although these changes may seem to make the problem easier, they may significantly alter the distribution of observations revealed to the algorithm, breaking statistical assumptions made about the observations. Semirandom models are thus a means of penalizing brittle algorithms that are over-tuned to particular random models. Semirandom models are no easier than their random counterparts, as the adversary may opt to make no changes. 

Following \cite{fk}, we define the {\bf semirandom planted partition model} as follows. A graph is first generated according to the planted partition model, and then a monotone adversary may arbitrarily add edges within communities, and remove edges between communities. The semirandom model can simulate aspects of real-world graphs and other graph models, such as a wide range of degree or subgraph profiles, while ensuring that the true community structure remains present in the graph. Thus the semirandom model aims to capture the unpredictable nature of real data.

An algorithm is called {\bf robust} to monotone adversaries if, whenever it succeeds on a sample from the random model, it also succeeds after arbitrary monotone changes to that sample. Algorithms based on semidefinite programming (SDP) are typically robust, or can be modified to be robust, and essentially all known robust algorithms are based on convex programs such as SDPs. This property guarantees some ability of such algorithms to generalize to other models; by contrast, algorithms that over-exploit precise statistics of a random model will typically fail against similar random models with different statistics, and will also typically fail against a semirandom model.

\subsubsection{Regimes}

The stochastic block model admits two major regimes of study, along with other variants. The main distinction is between {\bf partial recovery}, in which the goal is only to recover a partition that is reliably correlated with the true partition better than random guessing, versus {\bf exact recovery}, in which the partition must be recovered perfectly. In partial recovery, one tends to take the within-community edge probability $p$ and the between-community edge probability $q$ to be $\Theta(1/n)$, whereas in exact recovery one takes them to be $\Theta(\log n / n)$. In these asymptotic regimes one observes a sharp threshold behavior: within some range of parameters $p$ and $q$, the problem is statistically impossible, and outside of that range one can find algorithms that succeed with high probability. For partial recovery, this is established in \cite{mns, mns2, massoulie}, and for exact recovery, the most general result on this threshold is established in \cite{as}.

This paper concerns the case of exact recovery. Thus we take probabilities $p = \tilde p \log n / n$, $q = \tilde q \log n / n$, and we suppose that the vertex set is partitioned into $r$ communities $S_1,\ldots,S_r$ of sizes $s_i = |S_i|$. We write $\pi_i = s_i / n$, the proportion of vertices lying in community $i$. More precisely, we only require $s_i/n \to \pi_i$ as $n \to \infty$; for instance, each vertex could be independently randomly assigned to a community, with community $i$ chosen with probability $\pi_i$. We hold $r$, $\tilde p$, $\tilde q$, and $\pi_i$ constant as $n \to \infty$ -- these are the parameters for the problem -- and we aim to exactly recover the communities $S_i$ (up to permutation) {\bf with high probability}, by which we mean probability $1 - o(1)$ as $n \to \infty$, within as broad a range of parameters as possible.

We specialize to the {\bf assortative} planted partition model, where $p > q$; some techniques may transfer to the {\bf dissortative} $p < q$ case, by judicious negations, but we do not elaborate on this.

There remains one more distinction to make: we have not yet specified whether the parameters $r$, $\tilde p$, $\tilde q$, and $\pi_i$ are hidden or known to the recovery algorithm. We present algorithms for two cases: the {\bf known sizes} case, where the community proportions $\pi_i$ are known but $\tilde p$ and $\tilde q$ are potentially unknown; and the {\bf unknown sizes} case, where the $\pi_i$ are unknown but $\tilde p$, $\tilde q$, and the number $r$ of communities are known. Even the case of fully unknown parameters is tractable \cite{as2}, but there are fundamental barriers that prevent robust approaches such as ours from achieving this. We elaborate on this in Section~\ref{section:robustness}.

\subsection{Contributions and prior work}

The main result of this paper is that two variants of a certain SDP achieve exact recovery, up to the information-theoretic threshold, against the planted partition model with multiple communities of potentially different sizes. These SDPs are furthermore robust to monotone adversaries.

There has been considerable prior work on the development of algorithms and lower bounds for the stochastic block model, with algorithms making use of a wide range of techniques; see for instance the introduction to \cite{abh} and works cited there. The use of semidefinite programming for exact recovery originated with \cite{fk}, who achieved robust exact recovery in the case of two communities of equal sizes, falling slightly short of the optimal performance threshold. More recently, semidefinite algorithms have been found especially effective on sparse graphs \cite{csx,abh}, and have been proven to match lower bounds for the planted partition model in the case of two equal-sized communities \cite{ban15,hwx}, two different-sized communities \cite{hwx2}, and multiple equal-sized communities \cite{hwx2,abkk}, but the case of recovering multiple communities of different sizes through SDPs remained unresolved until now. The SDP that we consider in this paper has appeared before in the literature \cite{gv,abkk}, and in fact Agarwal et al.\ \cite{abkk} conjectured that it achieves exact recovery up to the threshold; the main result of this paper resolves this conjecture affirmatively.

Abbe and Sandon \cite{as} established the information-theoretic threshold for the general stochastic block model with individual community-pair probabilities $Q_{ij}$, which we visit in Section~\ref{section:lower-bounds}. In addition to proving a sharp lower bound, they analyzed an algorithm that succeeds with high probability up to their lower bound, thus precisely determining the statistical limits for exact recovery. Their result may appear strictly more general than ours, applying to the general stochastic block model rather than the more specific planted partition model. However, their algorithm involves a highly-tuned spectral clustering step that is likely not robust to monotone adversaries or other forms of perturbation (see Section~\ref{section:robustness} for discussion), and so our work can be seen as an improvement from the robustness standpoint.

In fact, there are barriers against more general \emph{semirandom} results. We show in Section~\ref{section:robustness} that no algorithm robust to monotone adversaries can handle the case of fully unknown parameters, nor can it extend from the planted partition model to the slightly more general strongly assortative block model. Thus our SDP-based approach is already essentially as general as one could hope for without compromising the strength of its robustness.

Although we are primarily interested in the exact recovery regime, we will take a brief detour and survey the parallel line of work on partial recovery. In contrast to exact recovery, none of the partial recovery algorithms that succeed down to the information-theoretic threshold are based on SDPs, and they have no robustness guarantees. There are robust SDP methods known \cite{gv,mpw,mmv} but they fall short of the threshold; this is essential, as it has been shown that no algorithm can robustly achieve the partial recovery threshold \cite{mpw}. Although robust algorithms for partial recovery cannot reach the threshold, there is an SDP that is conjectured to come quite close \cite{jmr} and is known to achieve the threshold in the limit of large average degree \cite{ms}; however, this analysis does not come with a robustness guarantee. By contrast, this paper shows that robust recovery is efficiently achievable up to the threshold in the exact recovery setting.

Community detection in semirandom models was previously discussed in \cite{fk,csx2,al}, and further since the first appearance of this paper, in \cite{mpw,mmv,hwx}. Other works \cite{cl,mmv} discuss robustness to some quantity of corruptions or outlier nodes.

\subsection{Overview of techniques}

In this section we summarize the arguments by which we prove that our SDPs achieve exact recovery against the semirandom planted partition model. More precisely, we show that with high probability, the unique SDP optimum is a matrix that exactly encodes the correct community structure; in contrast to some SDP-based algorithms (both for community detection and otherwise), there is no `rounding' or post-processing procedure required.

The first step is to show robustness in the sense that if the SDP succeeds against a particular graph, then it will continue to succeed if that graph is modified by monotone changes. This follows from a simple argument in \cite{fk} which essentially observes that the SDP optimizes over a space of solutions, and monotone changes improve the objective value of the true solution more than they improve the objective value of any other solution.

In light of the above argument, it suffices to show that our SDP succeeds with high probability against the random model. As in previous work on SDPs for exact recovery, our proof proceeds by constructing a dual certificate. The idea here is that the SDP is a maximization problem for which the true solution is feasible; what we need to show is that no other solution has a larger objective value than the true one, which can be done by finding a matching solution to the dual SDP. This ``dual certificate'' that we construct depends on the random graph and can be shown to be dual-feasible with high probability. As is typical, the construction of the dual certificate is guided by complementary slackness, which provides a set of necessary conditions. However, these necessary conditions are not enough to uniquely determine the dual certificate and so some creativity is necessary here in order to find an optimal certificate that gets all the way to the threshold. Part of what makes the general problem harder than the special cases considered previously \cite{hwx2} is that the general case has more dual variables that are not automatically determined by complementary slackness but that nonetheless need to be chosen carefully. The crucial step in the construction of our dual certificate is that by making a change of variables we are able to find a connection between the complementary slackness conditions and the non-negativity of differences of certain binomial random variables, which in turn are closely related to the information-theoretic threshold. It then becomes clear which parts of the dual solution need to crucially be set a particular way, and which ones have ``wiggle room.'' As is typical, showing that our dual matrix is positive semidefinite relies in part on the spectral concentration of the adjacency matrix, e.g.\ Theorem~5.2 in \cite{lr}.

\subsection{Organization of this paper}

In Section~\ref{section:algorithms} we present our two closely-related semidefinite programs for exact recovery. In Section~\ref{section:lower-bounds} we discuss the information-theoretic threshold determined in \cite{as}. In Section~\ref{section:robustness} we show that our SDPs are robust to the semirandom model, and we give some impossibility results for robust algorithms. In Section~\ref{section:proof} we prove our main result: our SDPs achieve exact recovery up to the information-theoretic threshold.

\section{Semidefinite algorithms and results}
\label{section:algorithms}

In this section we derive semidefinite programs for exact recovery, by taking convex relaxations of maximum likelihood estimators. 
Throughout we will use the letters $u,v$ for vertices and the letters $i,j$ for communities.

\subsection{Maximum likelihood estimators}
\label{sec:MLE}
Given an observed $n$-vertex graph, a natural statistical approach to recovering the community structure is to compute a maximum-likelihood estimate (MLE). We begin by stating the log-likelihood of a candidate partition into communities:
$$ \log \likeli \hspace{.3em} = \hspace{-1em}
\sum_{\substack{u     \sim v \\ \text{same community}}} \hspace{-1.5em} \log p \hspace{.5em} + \hspace{-1em}
\sum_{\substack{u     \sim v \\ \text{diff community}}} \hspace{-1em} \log q \hspace{.5em} + \hspace{-1em}
\sum_{\substack{u \not\sim v \\ \text{same community}}} \hspace{-1.5em} \log (1-p) \hspace{.5em} + \hspace{-1em}
\sum_{\substack{u \not\sim v \\ \text{diff community}}} \hspace{-1em} \log (1-q) $$
Here $\sim$ denotes adjacency in the observed graph. We can represent a partition by its $n \times n$ {\bf partition matrix}
$$ X_{uv} = \begin{cases} 1 & \text{if $u$ and $v$ are in the same community,} \\ 0 & \text{otherwise.} \end{cases} $$
In terms of $X$ and the observed $(0,1)$-adjacency matrix $A$, we can write
\begin{align*}
2 \log \likeli(X) &= \langle A,X \rangle \log p + \langle A,J-X \rangle \log q \\
    &+ \langle J-I-A, X \rangle \log (1-p) + \langle J-I-A, J-X \rangle \log (1-q)
\end{align*}
where $I$ is the identity matrix, $J$ is the all-ones matrix, and $\langle \cdot,\cdot \rangle$ denotes the Frobenius (entry-wise) inner product of matrices.
Expanding, and discarding terms that do not depend on $X$, including $\langle X,I \rangle = n$:
\begin{align*}
2 \log \likeli(X) + \text{const} &= \alpha \langle A,X \rangle - \beta \langle J,X \rangle \\
& = \alpha \langle A,X \rangle - \beta \sum_{i=1}^r s_i^2,
\end{align*}
where
$$ \alpha = \log \frac{p(1-q)}{q(1-p)}, \quad \beta = \log \frac{1-q}{1-p}. $$
The assumption $p > q$ implies that $\alpha$ and $\beta$ are positive.

At this stage it is worth distinguishing the two cases of known and unknown sizes.
In the first form, the block sizes $s_i$ are known; then the second term above is a constant, and the MLE amounts to a minimum fixed-sizes multisection problem, and is NP-hard for worst-case $A$.
\begin{program}[Known sizes MLE]
\begin{align*}
\text{maximize} \quad & \langle A,X \rangle \\
\text{over} \quad & \text{partition matrices $X$ with community sizes }s_1,\ldots,s_r.
\end{align*}
\end{program}

In the second form of the problem, the block sizes $s_i$ are unknown (but $r$ is known). Now the MLE requires knowledge of $p$ and $q$, and the resulting regularized minimum multisection problem is also likely computationally hard in general.
\begin{program}[Unknown sizes MLE]
\begin{align*}
\text{maximize} \quad & \langle A,X \rangle - \omega \langle J,X \rangle \\
\text{over} \quad & \text{all partition matrices }X \text{ on $r$ communities},
\end{align*}
where
\begin{equation}\label{eq:omega-definition}
\omega = \frac{\beta}{\alpha} = \frac{\log (1-q) - \log (1-p)}{\log p + \log(1-q) - \log q - \log(1-p)}.
\end{equation}
\end{program}

Although it is not obvious from the definition, we can think of $\omega$ as a sort of average of $p$ and $q$; the proof is deferred to Appendix~\ref{appendix:omega-bounds}:
\begin{lemma}
\label{lemma:omega-bounds}
For all $0 < q < p < 1$, we have $q < \omega < p$.
\end{lemma}

In order to compute either MLE in polynomial time, we will pass to a relaxation and show that the computation succeeds with high probability.

\subsection{Semidefinite algorithms}

The seminal work of \cite{gw} began a successful history of semidefinite relaxations for cut problems. Proceeding roughly in this vein to write a convex relaxation for the true feasible region of partition matrices of given sizes, one might reasonably arrive at the following relaxation of the known-sizes MLE:

\begin{program}[Known sizes, primal, weak form]
\label{prog:known-primal-weak}
\begin{align*}
\text{maximize} \quad & \langle A,X \rangle \\
\text{subject to} \quad &\langle J,X \rangle = \sum_i s_i^2, \\
\forall u \quad & X_{uu} = 1, \\
    & X \geq 0, \\
    & X \succeq 0.
\end{align*}
\end{program}
Here $X \geq 0$ indicates that $X$ is entry-wise nonnegative, and $X \succeq 0$ indicates that $X$ is positive semidefinite (and, in particular, symmetric).

This SDP appears in \cite{al} (under the name SDP-2'). Under the assumption of equal-sized communities, a stronger form involving row sum constraints appears in \cite{al, hwx2, abkk}, and the latter two papers find that this strengthening achieves exact recovery up to the information-theoretic lower bound in the case of equal-sized communities. 

In the case of unequal community sizes, it is more difficult to pursue this line of strengthening. The authors have analyzed the weaker SDP above for multiple unbalanced communities, and found that it does achieve exact recovery within some parameter range in the logarithmic regime, but its threshold for exact recovery is strictly worse than the information-theoretic threshold. (We have opted to omit these proofs from this paper, in view of our stronger results on the programs below.)

Instead, we revisit the somewhat arbitrary decision to encode the partition matrix with entries $0$ and $1$. Indeed, SDPs for the unbalanced two-community case tend to use entries $-1$ and $1$ \cite{fk,abh,hwx2}, with success up to the information-theoretic lower bound \cite{hwx2}. Some choices of entry values will result in a non-PSD partition matrix, so we opt for the choice of entries for which the partition matrix is only barely PSD: namely, we define the {\bf centered partition matrix}
$$ X_{uv} = \begin{cases} 1 & \text{if $u$ and $v$ are in the same community,} \\ \frac{-1}{r-1} & \text{otherwise.} \end{cases} $$
Recall that $r$ is the number of communities. This matrix is PSD as it is a Gram matrix: if we geometrically assign to each community a vector pointing to a different vertex of a centered regular simplex in $\RR^{r-1}$, the resulting Gram matrix is precisely the centered partition matrix.

Aiming to recover this matrix, we reformulate our SDP as follows:
\begin{program}[Known sizes, primal, strong form]
\label{prog:known-primal}
\begin{align*}
\text{maximize} \quad & \langle A,X \rangle \\
\text{subject to} \quad &\langle J,X \rangle = \frac{r}{r-1} \sum_i s_i^2 - \frac{1}{r-1} \, n^2, \\
\forall u \quad & X_{uu} = 1, \\
    & X \geq \frac{-1}{r-1}, \\
    & X \succeq 0.
\end{align*}
\end{program}
This SDP bears a strong similarity to classical SDPs for maximum $r$-cut \cite{fj}, and \cite{abkk} conjectured that it achieves exact recovery for unbalanced multisection up to the information-theoretic lower bound. Our main result resolves this conjecture affirmatively.

Other than the natural motivation of following classical approaches to $r$-cut problems, the change from partition matrices to centered partition matrices buys us something mathematically concrete: the intended primal solution now has rank $r-1$ instead of rank $r$, which through complementary slackness will entail one less constraint on a candidate dual optimum.

We can write down a very similar relaxation for the MLE in the case of unknown sizes but known $p$ and $q$:
\begin{program}[Unknown sizes, primal, strong form]
\label{prog:unknown-primal}
\begin{align*}
\text{maximize} \quad & \langle A,X \rangle - \omega \langle J,X \rangle \\
\text{subject to} \quad \forall u \quad & X_{uu} = 1, \\
    & X \geq \frac{-1}{r-1}, \\
    & X \succeq 0.
\end{align*}
\end{program}
Here $\omega$ is as defined in (\ref{eq:omega-definition}). Our main assertion is that these SDPs achieve exact recovery up to the lower bounds in \cite{as}:
\begin{theorem}\label{thm:recovery}
Given input from the planted partition model with parameters $(\tilde p, \tilde q, \pi)$, Programs~\ref{prog:known-primal} (known sizes) and \ref{prog:unknown-primal} (unknown sizes) recover the true centered partition matrix as the unique optimum, with probability $1 - o(1)$ over the random graph, within the information-theoretically feasible range of parameters described by Theorem~\ref{thm:as-lower-bounds}.
\end{theorem}

To reiterate, the known sizes SDP (Program~\ref{prog:known-primal}) requires knowledge of the sum of squares of community sizes (equivalently $\sum_i \pi_i^2$) and the number $r$ of communities, but $p$ and $q$ can be unknown; the unknown sizes SDP (Program~\ref{prog:unknown-primal}) requires knowledge of $p$, $q$, and $r$ (or at least $\omega$ and $r$), but the community proportions $\pi$ can be unknown.

Program~\ref{prog:unknown-primal} can tolerate some mis-specification in the value of $\omega$, though this tolerance may shrink to zero as the problem parameters approach the information-theoretic threshold. This tolerance is implicit in the proof of Theorem~\ref{thm:recovery}.

\section{Lower bounds}
\label{section:lower-bounds}

In this section we visit the general information-theoretic lower bounds established by \cite{as}, and we specialize them to the planted partition model, recovering bounds that directly generalize those of \cite{hwx2} for the case of two unequal communities. Meanwhile, we recall operational interpretations of the lower bounds, which will be our main concrete handle on them.

Consider the general stochastic block model (probabilities $Q_{ij}$) in the regime where community sizes are $s_i = \pi_i n$ and probabilities are $Q_{ij} = \tilde Q_{ij} \log n / n$, with $\pi$ and $\tilde Q$ constant as $n \to \infty$. The threshold for exact recovery is as follows:
\begin{theorem}[\cite{as}]
\label{thm:as-lower-bounds}
For $i \ne j$ define the {\bf Chernoff--Hellinger divergence}
\begin{equation}\label{eq:ch-divergence}
D_+(i,j) \stackrel{\mathrm{def}}{=} \sup_{t \in [0,1]} \sum_k \pi_k (t\, \tilde Q_{ik} + (1-t) \tilde Q_{jk} - \tilde Q_{ik}^t \tilde Q_{jk}^{1-t}).
\end{equation}

Then exact recovery is information-theoretically solvable if
\begin{equation}\label{eq:divergence-condition}
\forall i \neq j \quad D_+(i,j) > 1.
\end{equation}
Within this range, there exist efficient algorithms that succeed with probability $1 - o(1)$.

Conversely, exact recovery is impossible if there exists a pair $(i,j)$ with $D_+(i,j) < 1$. Within this range, any estimator for the community structure (even knowing the model parameters) will fail with probability $1 - o(1)$.
\end{theorem}
\cite{as} also shows that the borderline case $D_+(i,j) = 1$ remains solvable; for simplicity we neglect this case.

As mentioned in \cite{as}, if we specialize as far as the planted partition model with two equal-sized communities, then the CH-divergence is obtained at $t=\frac12$, and one recovers the threshold $\sqrt{\tilde p} - \sqrt{\tilde q} \geq \sqrt{2}$, as seen in \cite{abh} and other works. In Appendix~\ref{appendix:lower-bound-closed-form}, we compute the following generalization:
\begin{proposition}
\label{prop:lower-bound-closed-form}
For the planted partition model with parameters $(\tilde p, \tilde q, \pi)$, the CH-divergence is given by
$$ D_+(i, j) = \pi_i \tilde q + \pi_j \tilde p - \gamma + \frac12 \tau (\pi_i - \pi_j) \log\left( \frac{\pi_j \tilde p}{\pi_i \tilde q} \cdot \frac{\tau(\pi_i-\pi_j)+\gamma}{\tau(\pi_j-\pi_i)+\gamma} \right), $$
$$ \gamma = \sqrt{\tau^2(\pi_i-\pi_j)^2 + 4 \pi_i \pi_j \tilde p \tilde q}, $$
$$ \tau = \frac{\tilde p - \tilde q}{\log \tilde p - \log \tilde q}. $$
\end{proposition}

Note that in our regime (logarithmic average degree), we have
\begin{equation}\label{eq:omega-asymptotics}
\omega = \tau \log n / n + o(\log n / n).
\end{equation}

The divergence expression in Proposition~\ref{prop:lower-bound-closed-form} strongly resembles the lower bound proven in \cite{hwx2} for the case of two communities of different sizes. Indeed, in the notation of Lemma~2 of \cite{hwx2}, we recognize our expression as
$$ g\left( \pi_i,\pi_j,\tilde p,\tilde q,\tau(\pi_j-\pi_i) \right). $$
From that lemma, the following operational definition of the CH-divergence is immediate for the planted partition model:
\begin{lemma}
\label{lemma:divergence-binomial}
Let $i \neq j$ be communities, and let $v$ be a vertex in community $i$. Let $E(v,j)$ denote the number of edges from $v$ to vertices in community $j$. Suppose that $T(n) = \tau (\pi_i - \pi_j) \log n + o(\log n)$. The probability of the tail event $E(v,i) - E(v,j) \leq T(n)$ is $n^{-D_+(i,j) + o(1)}$.
\end{lemma}
By a naive union bound, when $D_+(i,j) > 1$ for all pairs $i, j$, then we can assert with high probability that none of these tail events occur, over all vertices and communities.

A similar operational interpretation is given directly in \cite{as}, phrased in terms of hypothesis testing between multivariate Poisson distributions. This result keeps more complete track of the $o(1)$ term, so as to guarantee the union bound even when $D_+(i,j) = 1$.

Lastly we note the following monotonicity property of the divergence, with a proof deferred to Appendix~\ref{appendix:divergence-monotonicity}:
\begin{proposition}
\label{prop:divergence-monotonicity}
In the planted partition model, the CH-divergence $D_+(i,j)$ is monotone increasing in $\pi_i$ and $\pi_j$ (for any fixed $\tilde p, \tilde q$).
\end{proposition}
Thus, when determining whether exact recovery is feasible in the planted partition model for some set of parameters, it suffices to check the CH-divergence between the two smallest communities.

\section{Semirandom robustness and its consequences}
\label{section:robustness}

\subsection{The semirandom model}

Let $\hat X$ be the ground truth partition of the vertices into communities.
\begin{definition}
Following \cite{fk} we define a {\bf monotone adversary} to be a process which takes as input a graph (for instance, a random graph sampled from the stochastic block model with ground truth $\hat X$) and makes any number of {\bf monotone changes} of the following types:
\begin{itemize}
\item The adversary can add an edge between two vertices in the same community of $\hat X$.
\item The adversary can remove an edge between two vertices in different communities of $\hat X$.
\end{itemize}
\end{definition}
These monotone changes appear to only strengthen the presence of the true community structure $\hat X$ in the observed graph, yet they may destroy statistical properties of the random model. The semirandom model is designed to penalize brittle algorithms that over-rely on specific stochastic models. It does not intend to mimic any real-world adversarial scenario, but it does intend to model the inherent unpredictability of real-world data.

It may help to consider examples of how such an adversary could break an algorithm:
\begin{itemize}
\item Many algorithms perform PCA on the adjacency matrix \cite{lr}. The adversary could plant a slightly denser sub-community structure in a community, thus splitting one cluster of vertices into several nearby sub-clusters in the PCA, and introducing doubt as to which granularity of clustering is appropriate.
\item An adversary could introduce a noise distribution that changes the shape of clusters of vertices in the PCA or spreads them out, resulting in either a failure to cluster correctly, or else a failure in subsequent steps of estimating parameters and improving the community structure (as in \cite{as2}).
\end{itemize}

These are extreme examples, but they correspond to realistic concerns:
\begin{itemize}
\item Real community structure is sometimes hierarchical, e.g.\ tight friend groups within a larger social community.
\item Many real networks have hubs, or a degree distribution that is nowhere near Gaussian, so hypothesis tests designed for distributions of roughly Gaussian shape may have trouble generalizing.
\end{itemize}

\subsection{Robustness}

In this section, we establish that our SDPs are robust to monotone adversaries. We first elaborate on the definitions discussed in the introduction.

\begin{definition}
Suppose $f$ is a (deterministic) algorithm for recovery in the stochastic block model, namely $f$ takes in an adjacency matrix $A$ and outputs a partition $f(A)$ of the vertices. We say $f$ is {\bf robust to monotone adversaries} if: for any $A$ such that $f(A) = \hat X$, we have $f(A') = \hat X$ for any $A'$ obtained from $A$ via a sequence of monotone changes.
\end{definition}
We modify this definition slightly for SDPs in order to deal with the fact that an SDP may not have a unique optimum. By abuse of notation, let $\hat X$ also refer to the centered partition matrix corresponding to the partition $\hat X$.
\begin{definition}
Suppose $P_A$ is a semidefinite program which depends on the adjacency matrix $A$ (for instance, Program \ref{prog:known-primal} or Program \ref{prog:unknown-primal}). We say $P_A$ is {\bf robust to monotone adversaries} if: for any $A$ such that $\hat X$ is the unique optimum to $P_A$, we have that $\hat X$ is the unique optimum to $P_{A'}$ for any $A'$ obtained from $A$ via a sequence of monotone changes.
\end{definition}

SDPs tend to possess such robustness properties. We will now show that our SDPs are no exception, following roughly the same type of argument as \cite{fk}.

\begin{proposition}\label{prop:robustness}
Programs \ref{prog:known-primal} (known sizes) and \ref{prog:unknown-primal} (unknown sizes) are robust to monotone adversaries.
\end{proposition}
\begin{proof}
Let $P_A$ be either Program \ref{prog:known-primal} or Program \ref{prog:unknown-primal} (the proof is identical for both cases). Suppose the true centered partition matrix $\hat X$ is the unique optimum for $P_A$. By induction it is sufficient to show that $\hat X$ is the unique optimum for $P_{A'}$ where $A'$ is obtained from $A$ via a single monotone change. Note that $P_A$ and $P_{A'}$ have the same feasible region because $A$ only affects the objective function. Let $P_A(X)$ denote the objective value of a candidate solution $X$ for $P_A$, namely $P_A(X)$ is $\langle A,X \rangle$ for Program \ref{prog:known-primal} and $\langle A,X \rangle - \omega \langle J,X \rangle$ for Program \ref{prog:unknown-primal}. First consider the case where $A'$ is obtained from $A$ via a single monotone edge-addition step. Since the added edge lies within a community of $\hat X$ we have $P_{A'}(\hat X) = P_A(\hat X) + 2$. For any matrix $X$ feasible for $P_A$ (equivalently, feasible for $P_{A'}$), we have $P_{A'}(X) \le P_A(X) + 2$; this follows from $X \le 1$ (entry-wise), which is implied by the constraints $X_{vv} = 1$ and $X \succeq 0$. If $X \ne \hat X$ we have $P_{A'}(\hat X) = P_A(\hat X) + 2 > P_A(X) + 2 \ge P_{A'}(X)$ and so $\hat X$ is the unique optimum of $P_{A'}$. Similarly, for the case where $A'$ is obtained from $A$ via a single monotone edge-removal, we have $P_{A'}(\hat X) = P_A(\hat X) + \frac{2}{r-1}$ and $P_{A'}(X) \le P_A(X) + \frac{2}{r-1}$ (using the constraint $X \ge \frac{-1}{r-1}$) and the result follows.
\end{proof}

Recall that in the semirandom planted partition model, a random graph is generated according to the random (planted partition) model and then a monotone adversary is allowed to make monotone changes. Once we have established our main result (Theorem~\ref{thm:recovery}) on the success of our SDPs against the random model, it is an immediate corollary of robustness that our SDPs also succeed against the semirandom model.

\begin{proposition}\label{prop:semirandom-recovery}
Programs~\ref{prog:known-primal} (known sizes) and~\ref{prog:unknown-primal} (unknown sizes) achieve exact recovery against the \emph{semirandom} planted partition model, with probability $1 - o(1)$, up to the information-theoretic threshold for the random model (given in Theorem~\ref{thm:as-lower-bounds}).
\end{proposition}

\subsection{BM-ordering and strongly assortative block models}

\newcommand{\SBM}{\mathrm{SBM}}
\newcommand{\PPM}{\mathrm{PPM}}

Let $\SBM(n,\tilde Q,\pi)$ and $\PPM(n,\tilde p,\tilde q,\pi)$ denote the stochastic block model and planted partition model, respectively. Given input from a stochastic block model $\SBM(n,\tilde Q,\pi)$, we can simulate certain other stochastic block models $\SBM(n,\tilde Q',\pi)$, by adding edges within communities independently at random, and likewise removing edges between communities. Specifically, we can simulate any block model for which $\tilde Q_{ii}' \ge \tilde Q_{ii}$ and $\tilde Q_{ij}' \le \tilde Q_{ij}$, for all communities $i \neq j$; in this case, following \cite{al}, we say that $\SBM(n,\tilde Q',\pi)$ dominates $\SBM(n,\tilde Q,\pi)$ in block model ordering ({\bf BM-ordering}).

In the case when the original model is a planted partition model, this simulation step can be thought of as a specific monotone adversary. Block models that dominate a planted partition model will fall among the {\bf strongly assortative} class: those for which $\tilde Q_{ii} > \tilde Q_{jk}$ whenever $j \neq k$, i.e.\ all intra-community probabilities exceed all inter-community probabilities. The following is immediate from Proposition~\ref{prop:semirandom-recovery}:
\begin{proposition}
Programs~\ref{prog:known-primal} and~\ref{prog:unknown-primal} achieve exact recovery with high probability against any strongly assortative block model that dominates a planted partition model lying within the information-theoretically feasible range. 
\end{proposition}
Extrapolating slightly, we can assert that this extension to the strongly assortative block model tends to be automatic for semidefinite approaches because they tend to be robust. This has been previously noted by \cite{al} through essentially the same arguments. By taking $\tilde p = \min_i \tilde Q_{ii}$ and $\tilde q = \max_{j \neq k} \tilde Q_{jk}$, we could obtain from Proposition~\ref{prop:lower-bound-closed-form} a more explicit description of this sufficient condition for exact recovery.

\subsection{Difficulties with general block models}

Most natural SDPs tend to be robust to monotone adversaries. This strength of semidefinite approaches -- their ability to adapt to other random models following BM-ordering -- can be used to also reveal their limitations. We will show in this section that it is impossible for an algorithm robust to monotone changes to match the information-theoretic lower bound of \cite{as} in general, and even for strongly assortative block models.

As an example, define
$$ \tilde Q_1 = \left( \begin{array}{ccc}
a            & b & c + \epsilon \\
b            & a & c \\
c + \epsilon & c & a \end{array} \right), \qquad
\tilde Q_2 = \left( \begin{array}{ccc}
a & b & c \\
b & a & c \\
c & c & a \end{array} \right). $$
For a suitable setting of $\epsilon > 0$ and $a,b,c,\pi$, it is possible for $\SBM(n,\tilde Q_1,\pi)$ to be information-theoretically feasible for exact recovery while $\SBM(n,\tilde Q_2,\pi)$ is not. The following values provide an explicit example:
$$ a = 31.4,\quad b = 15,\quad c = 10,\quad \epsilon = 1,\quad \pi = (1/3,1/3,1/3). $$
However, $\tilde Q_1$ dominates $\tilde Q_2$ in BM-ordering. As our SDPs cannot hope to recover against $\SBM(n,\tilde Q_2,\pi)$, it follows that they fail also against $\SBM(n,\tilde Q_1,\pi)$, even though this model is information-theoretically feasible.

This example shows how monotone changes become subtly unhelpful in the strongly assortative block model: the first two communities become harder to distinguish under these monotone changes because their interactions with the third community become more similar. Arguably this makes the semirandom model inappropriate for such a general block model. Nonetheless, monotone robustness is a property of our SDPs, and also of all prior SDPs in the community detection literature (at least after minor strengthening), and so the limitations below apply at least to these specific SDPs. Thus we are able to learn about the limitations of these SDPs by studying their robustness properties.

The argument above applies to any algorithm that is robust to the semirandom model; this means no robust algorithm can achieve the threshold. This motivates us to conjecture a different ``monotone threshold'' for general block models, which we believe captures the information-theoretic limits in the general semirandom model. Define the {\bf monotone divergence}
$$ D_+^m(i,j) = \sup_{t \in [0,1]} \sum_{k \in \{i,j\}} \pi_k (t\, \tilde Q_{ik} + (1-t) \tilde Q_{jk} - \tilde Q_{ik}^t \tilde Q_{jk}^{1-t}). $$

Note that $D_+^m(i,j)$ is simply the value of $D_+(i,j)$ after setting $Q_{ik} = Q_{jk}$ for all $k \notin \{i,j\}$; this is a change in model that the monotone adversary can simulate (for instance set $Q_{ik} = Q_{jk} = 0$), and it is in fact the best change-in-model that the adversary can simulate if it wants to decrease $D_+(i,j)$ as much as possible for a specific $(i,j)$ pair. It follows that if $D_+^m(i,j) < 1$ for some $i \ne j$ then there does not exist a robust algorithm achieving exact recovery. We conjecture that conversely, if $D_+^m(i,j) \ge 1$ for all $i \ne j$, and if the block model is furthermore {\bf weakly assortative} ($Q_{ii} > Q_{ij}$ for all $i \neq j$), then there exists a robust algorithm achieving exact recovery against this block model.

\subsection{Difficulties with unknown parameters}

Non-semidefinite techniques in \cite{as2} achieve exact recovery up to the threshold without knowing any of the model parameters. One might ask whether it is possible for a robust algorithm (such as our SDPs) to achieve this; we now argue that this is not possible in general even in the planted partition model.

Consider for example a strongly assortative block model $\SBM(n,\tilde Q,\pi)$, on four communities, where
$$ \tilde Q = \left( \begin{array}{cccc}
a & b & c & c \\
b & a & c & c \\
c & c & a & b \\
c & c & b & a
\end{array} \right) $$
and $a > b > c$. This model may be simulated by a monotone adversary acting on either of the planted partition models $\PPM(n,a,b,\pi)$ and $\PPM(n,b,c,\pi')$, where $\pi' = (\pi_1 + \pi_2, \pi_3 + \pi_4)$. 
Suppose that we had a robust algorithm for exact recovery in the planted partition model without knowing any of the parameters. For a suitable setting of $a,b,c,\pi$, such an algorithm should be able to achieve exact recovery against the two planted partition models listed above. By robustness, it will still recover the same partition, with high probability, when presented with the strongly assortative block model $\SBM(n,\tilde Q,\pi)$. But now we have a contradiction: the algorithm allegedly recovers both partitions (corresponding to $\pi$ and $\pi'$) with high probability.

In effect, these two planted partition models have zero ``monotone total variation distance'', though we do not formalize this notion here. It is necessary to know some model parameters in advance in order for robust algorithms to distinguish such models. A few approaches are available to overcome this drawback:
\begin{itemize}

\item One could statistically estimate some or all of the parameters before running the SDP, as in Appendix~B of \cite{hwx2}. However, this statistical approach relies on the specific random model and spoils our robustness guarantees.

\item One could try running the SDP several times on a range of possible input parameters, ignoring any returned solutions that are not partition matrices. A close reading of Section~\ref{section:proof} reveals that, when running Program~\ref{prog:unknown-primal} (unknown sizes), mis-guessing the parameter $\omega$ by any $1-o(1)$ factor does not affect whether one succeeds with high probability.

This approach may return several valid solutions. In the example above, this approach will recover both of the given planted partition models, with high probability.
In general this approach recovers the type of hierarchical community structure that the above example exhibits.

\end{itemize}

\section{Proof of exact recovery}
\label{section:proof}

In this section we prove our main result (Theorem~\ref{thm:recovery}) which states that our SDPs achieve exact recovery against the planted partition model, up to the information-theoretic limit. Specifically, we show that if the divergence condition (\ref{eq:divergence-condition}) holds, then with high probability, the true centered partition matrix $\hat X$ is the unique optimum for our SDPs. The main idea of the proof is to construct a solution to the dual SDP in order to bound the value of the primal.

\subsection{Notation}
Recall that we use the letters $u,v$ for vertices and the letters $i,j$ for communities. We let $\one$ denote the all-ones vector, $\one_i$ denote the indicator vector of $S_i$, $I$ denote the identity matrix, and $J$ denote the all-ones matrix. When $M$ is any matrix, $M_{S_i S_j}$ will denote the submatrix indexed by $S_i \times S_j$, and we abbreviate $M_{S_i S_i}$ by $M_{S_i}$.

Let $A$ be the adjacency matrix of the observed graph and write $E(i,j) = \one_i^\top A \one_j$; when $i \ne j$ this is the number of edges between communities $i$ and $j$, and when $i = j$ this is twice the number of edges within community $i$.

All asymptotic notation refers to the limit $n \to \infty$, with parameters $\tilde p$, $\tilde q$, and $\pi$ held fixed. Throughout, we say an event occurs ``with high probability'' if its probability is $1 - o(1)$.

\subsection{Weak duality}

Following the standard dual certificate approach, we begin by writing down duality and complementary slackness for our SDP. This will lead us to a set of sufficient conditions, outlined in Proposition~\ref{prop:dual-cert} below.

We first write down the dual of Program~\ref{prog:unknown-primal}:
\begin{program}[Dual, unknown sizes]
\label{prog:unknown-dual}
\begin{align*}
\text{minimize} \quad & \sum_v \nu_v + \frac{1}{r-1} \langle J, \Gamma \rangle \\
\text{subject to} \quad & \Lambda \stackrel{\mathrm{def}}{=} \diag(\nu) + \omega J - A - \Gamma \succeq 0, \\
    & \Gamma \ge 0 \quad \text{symmetric}.
\end{align*}
\end{program}
Here the $n$-vector $\nu$ and the $n \times n$ matrix $\Gamma$ (both indexed by vertices) are dual variables, and $\omega$ is as defined in (\ref{eq:omega-definition}). We can now state weak duality in this context:
\begin{align*}
\langle A,X \rangle - \omega \langle J,X \rangle &= \langle A - \omega J, X \rangle = \langle \diag(\nu) - \Gamma - \Lambda, X \rangle \\
& = \sum_v \nu_v - \langle \Gamma, X \rangle - \langle \Lambda, X \rangle \\
& = \sum_v \nu_v + \frac{1}{r-1} \langle J, \Gamma \rangle - \langle \Gamma, X+\frac{1}{r-1} J \rangle - \langle \Lambda, X \rangle.
\end{align*}

This implies weak duality $\langle A,X \rangle - \omega \langle J,X \rangle \le \sum_v \nu_v + \frac{1}{r-1} \langle J, \Gamma \rangle$ (the primal objective value is at most the dual objective value) because $\langle \Gamma, X+\frac{1}{r-1}J \rangle \ge 0$ (since $\Gamma \ge 0$, $X + \frac{1}{r-1} J \ge 0$) and $\langle \Lambda, X \rangle \ge 0$ (since $\Lambda \succeq 0, X \succeq 0$).

\subsection{Complementary slackness}

From above we have the following complementary slackness conditions. If $X$ is primal feasible and $(\nu,\Gamma)$ is dual feasible then $X$ and $(\nu,\Gamma)$ have the same objective value if and only if $\langle \Gamma, X+\frac{1}{r-1}J \rangle = 0$ and $\langle \Lambda, X \rangle = 0$. Since $\Lambda$ and $X$ are PSD, $\langle \Lambda, X \rangle = 0$ is equivalent to $\Lambda X = 0$ (this can be shown using the rank-1 decomposition of PSD matrices), which in turn is equivalent to $\mathrm{colspan}(X) \subseteq \ker(\Lambda)$.

Although we have only considered Program~\ref{prog:unknown-primal} so far, everything we have done also applies to Program~\ref{prog:known-primal}. The dual of Program~\ref{prog:known-primal} is identical to Program~\ref{prog:unknown-dual}, except that $\omega$ is replaced by a dual variable, and there is a corresponding term in the objective. By deterministically choosing this dual variable to take the value $\omega$, we arrive at a dual program with the same feasible region and complementary slackness conditions as Program~\ref{prog:unknown-dual}. From this point onward, the same arguments apply to both Programs \ref{prog:known-primal} and \ref{prog:unknown-primal}.

Let $\hat X$ be the true centered partition matrix with $(1,\frac{-1}{r-1})$ entries. The following proposition gives a sufficient condition for $\hat X$ to be the unique optimum for Programs \ref{prog:known-primal} and \ref{prog:unknown-primal}.

\begin{proposition}
\label{prop:dual-cert}
Suppose there exists a dual solution $(\nu,\Gamma)$ satisfying:
\begin{itemize}
\item $\Lambda \succeq 0$,
\item $\Gamma_{S_i} = 0$ for all $i$,
\item $\Gamma_{S_i S_j} > 0$ (entry-wise) for all $i \ne j$,
\item $\ker(\Lambda) = \mathrm{span}\{\one_i - \one_j\}_{i,j}$.
\end{itemize}
Then $\hat X$ is the unique optimum for Programs \ref{prog:known-primal} and \ref{prog:unknown-primal}. (Here $\Lambda$ is defined as $\Lambda = \diag(\nu) + \omega J - A - \Gamma$ as in Program \ref{prog:unknown-dual}.)
\end{proposition}

\begin{proof}
The first three assumptions imply that $(\nu,\Gamma)$ is dual feasible. The column span of $\hat X$ is $\mathrm{span}\{r \one_i - \one\}_i = \mathrm{span}\{\one_i - \one_j\}_{i,j}$ so the fourth assumption is that $\mathrm{colspan}(\hat X) = \ker(\Lambda)$, which is one of our two complementary slackness conditions. The assumption $\Gamma_{S_i} = 0$ implies $\langle \Gamma, \hat X + \frac{1}{r-1}J \rangle = 0$ because $\hat X + \frac{1}{r-1}J$ is supported on the diagonal blocks. This is the other complementary slackness condition, so complementary slackness holds, certifying that $\hat X$ is primal optimal $(\nu,\Gamma)$ is dual optimal.

To show uniqueness, suppose $X$ is any optimal primal solution. By complementary slackness, $\mathrm{colspan}(X) \subseteq \ker(\Lambda) = \mathrm{span}\{\one_i - \one_j\}_{i,j}$ and $\langle \Gamma, X+\frac{1}{r-1}J \rangle = 0$. Since $\Gamma_{S_i S_j} > 0$, this means $X_{S_i S_j} = \frac{-1}{r-1} J$ for all $i \ne j$. But since every column of $X$ is in $\mathrm{span}\{\one_i - \one_j\}_{i,j}$, we must now have $X_{S_i} = J$ and so $X = \hat X$.
\end{proof}

We note that Proposition~\ref{prop:dual-cert} is not novel in that all the arguments we have made so far are standard in the dual certificate approach.


\subsection{Construction of dual certificate -- overview}

We now explore the space of dual certificates that will satisfy the conditions of Proposition~\ref{prop:dual-cert}, so as to sound out how to construct such a certificate. The main result of this section is to rewrite the problem in terms of a new set of variables $\gamma_u$. We believe this change of variables is novel, and it is crucial to our approach because it will allow us to make a connection between between complementary slackness and certain differences of binomial variables that are closely related to the information-theoretic threshold (see Lemma~\ref{lemma:divergence-binomial}).

The condition $\mathrm{span}\{\one_i - \one_j\} \subseteq \ker(\Lambda)$ is equivalent to
$$\forall u, \forall i \ne j \quad e_u^\top \Lambda (\one_i - \one_j) = 0.$$
Using the definition $\Lambda = \diag(\nu) + \omega J - A - \Gamma$, this can be rewritten as the two equations
\begin{equation}
\label{eq:cs1}
\forall i \ne j, \forall u \in S_i \quad \nu_u + \omega (s_i-s_j) - E(u,i) + E(u,j) + \sum_{v \in S_j} \Gamma_{uv} = 0
\end{equation}
and
\begin{equation}
\label{eq:cs2}
\forall i \ne j, \forall u \notin S_i, u \notin S_j \quad \omega(s_i - s_j) - E(u,i) + E(u,j) - \sum_{v \in S_i} \Gamma_{uv} + \sum_{v \in S_j} \Gamma_{uv} = 0.
\end{equation}

We can disregard the equations (\ref{eq:cs2}) because they are implied by the equations (\ref{eq:cs1}) via subtraction. From (\ref{eq:cs1}) we have that, for any fixed $u \in S_i$, the quantity $\omega s_j - E(u,j) - \sum_{v \in S_j} \Gamma_{uv}$ must be independent of $j$ (for $j \ne i$). Hence let us define
\begin{equation}
\label{eq:gamma}
\gamma_u = \omega s_j - E(u,j) - \sum_{v \in S_j} \Gamma_{uv} \quad \forall j \ne i \quad \text{for } u \in S_i.
\end{equation}
Rewrite (\ref{eq:cs1}) as
\begin{equation}
\label{eq:nu}
\nu_u = E(u,i) - \omega s_i + \gamma_u \quad \text{for } u \in S_i,
\end{equation}
and rewrite (\ref{eq:gamma}) as
\begin{equation}
\label{eq:R}
R_{uj} = \omega s_j - E(u,j) - \gamma_u \quad \text{for } u \notin S_j,
\end{equation}
where $R_{uj}$ is shorthand for the row sum $\sum_{v \in S_j} \Gamma_{uv}$. Since $\Gamma$ is symmetric, $R_{uj}$ must be equal to the column sum $\sum_{v \in S_j} \Gamma_{vu}$, so we need for any $i \ne j$,
$$\sum_{u \in S_i} R_{uj} = \sum_{v \in S_j} R_{vi},$$
or equivalently:
$$\sum_{u \in S_i} [\omega s_j - E(u,j) - \gamma_u] = \sum_{v \in S_j} [\omega s_i - E(v,i) - \gamma_v],$$
or equivalently:
$$\omega s_i s_j - E(i,j) - \sum_{u \in S_i} \gamma_u = \omega s_i s_j - E(i,j) - \sum_{v \in S_j} \gamma_v,$$
or equivalently, there needs to exists a constant $c$ such that
\begin{equation}
\label{eq:c}
\sum_{u \in S_i} \gamma_u = c \quad \forall i.
\end{equation}
To recap, it remains to do the following. First choose $c$. Then choose $\gamma_u$ satisfying (\ref{eq:c}). Defining $R_{uj}$ by (\ref{eq:R}), we are now guaranteed that $\{R_{uj}\}_{u \in S_i}$ and $\{R_{vi}\}_{v \in S_j}$ are valid row and column sums respectively for $\Gamma_{S_i S_j}$ (for $i \ne j$). Then define $\nu_u$ by (\ref{eq:nu}), which guarantees $\mathrm{span}\{\one_i - \one_j\} \subseteq \ker(\Lambda)$. It remains to construct $\Gamma_{S_i S_j}$ explicitly from its row and column sums such that $\Gamma_{S_i S_j} > 0$. It also remains to show $\Lambda \succeq 0$ and $\ker(\Lambda) \subseteq \mathrm{span}\{\one_i - \one_j\}$. Note that we have not actually chosen any values for dual variables yet, other than what is required by complementary slackness; we have merely rewritten the complementary slackness conditions in terms of the new variables $\gamma_u$ and $R_{uj}$.

\subsection{Intervals for \texorpdfstring{$\gamma_v$}{gamma\_v}}

In this section we find necessary bounds for $\gamma_v$, which will guide our choice of these dual variables and of $c$. This is where the crucial connection between $\gamma_v$ and the information-theoretic threshold will become apparent. Let $v \in S_i$. For a lower bound on $\gamma_v$, we have that $\Lambda \succeq 0$ implies $\Lambda_{vv} \ge 0$ implies $\nu_v + \omega \ge 0$ which by (\ref{eq:nu}) implies $\gamma_v \ge \omega(s_i - 1) - E(v,i)$. For an upper bound, for any $j \ne i$ we must have that $\Gamma_{S_i S_j} > 0$ implies $R_{vj} > 0$ which by (\ref{eq:R}) implies $\gamma_v < \omega s_j - E(v,j)$. Therefore, $\gamma_v$ must lie in the interval
$$\gamma_v \in [\omega (s_i-1) - E(v,i) \; , \; \min_{j \ne i} \; \omega s_j - E(v,j)).$$
Our approach in choosing $\gamma_v$ will be to first make a preliminary guess $\gamma'_v$ and then add an adjustment term to ensure that (\ref{eq:c}) holds. In order to absorb this adjustment, we will aim for $\gamma'_v$ to lie in the slightly smaller interval
\begin{equation}
\label{eq:gamma-int}
\gamma_v' \in [\alpha_v \; , \; \beta_v]
\end{equation}
where
$$\alpha_v = \omega (s_i-1) - E(v,i) + \epsilon_1 \quad \text{and} \quad \beta_v = \min_{j \ne i} \; \omega s_j - E(v,j) - \epsilon_2 \quad \text{for } v \in S_i.$$
Here $\epsilon_1$ and $\epsilon_2$ are small $o(\log n)$ error terms which we will choose later.

The non-emptiness of these intervals is the crux of the proof, and provides the connection to the information-theoretic threshold:
\begin{lemma}
\label{lemma:intervals-nonempty}
If the divergence condition (\ref{eq:divergence-condition}) holds, then $\alpha_v < \beta_v$, for all $v$, with high probability.
\end{lemma}
\begin{proof}
For $v \in S_i$, we have
\begin{align*}
\beta_v - \alpha_v &= \min_{j \neq i} \left( (E(v,j) - E(v,i)) - \omega(s_j - s_i) - \epsilon_1 - \epsilon_2 \right) \\
&= \min_{j \neq i} (E(v,j) - E(v,i)) - \omega(s_j - s_i) - o(\log n) \\
&\stackrel{(\ref{eq:omega-asymptotics})}{=} \min_{j \neq i} (E(v,j) - E(v,i)) - \tau(\pi_j - \pi_i) \log n - o(\log n).
\end{align*}
By Lemma~\ref{lemma:divergence-binomial}, for each $v \in S_i$ and all $j \neq i$, the probability of the tail event
$$ (E(v,j) - E(v,i)) - \tau(\pi_j - \pi_i) \log n - o(\log n) \leq 0 $$
is $n^{-D_+(i,j) + o(1)}$. Thus, when $D_+(i,j) > 1$ for all pairs $(i,j)$, we can take a union bound over all $n(r-1)$ such events, to find that $\beta_v - \alpha_v \geq 0$ with probability $1 - o(1)$.
\end{proof}

By summing (\ref{eq:gamma-int}) over all $v \in S_i$ (roughly following (\ref{eq:c}), although $\gamma_v$ and $\gamma_v'$ are slightly different), we obtain a target interval for $c$:
\begin{equation}
\label{eq:c-int}
c \in [\alpha_i \; , \; \beta_i] \quad \forall i
\end{equation}
where
$$\alpha_i = \sum_{v \in S_i} \alpha_v = \omega s_i (s_i-1) - E(i,i) + s_i \epsilon_1$$
$$\beta_i = \sum_{v \in S_i} \beta_v = \sum_{v \in S_i} \min_{j \ne i} \; [\omega s_j - E(v,j) - \epsilon_2].$$

The endpoints of the interval (\ref{eq:c-int}) for $c$ will turn out to be highly concentrated near a pair of deterministic quantities, namely:
\begin{equation}
\label{eq:alpha-beta-bar}
\bar \alpha_i = (\omega-p) s_i (s_i - 1) + s_i \epsilon_1 \quad \text{and} \quad \bar \beta_i = (\omega-q) s_i s_{\mathrm{min} \ne i} - s_i \epsilon_2
\end{equation}
where $s_{\mathrm{min} \ne i} = \min_{j \ne i} s_j$.

\subsection{Choice of \texorpdfstring{$c$ and $\gamma_v$}{c and gamma\_v}}

We have now made the key insight that $c$ and $\gamma_v$ must lie in certain intervals in order to get all the way to the information-theoretic threshold. However, as long as we fulfill these requirements, we have some ``wiggle room'' in choosing the dual variables. We will make what seems to be the simplest choices. We can deterministically take
\begin{equation}\label{eq:cchoice}
c = \frac12 (\omega-q) s_{\mathrm{min}} s_{\mathrm{2ndmin}},
\end{equation}
where $s_{\mathrm{min}},s_{\mathrm{2ndmin}}$ are the sizes of the two smallest communities (which may be equal). Then for sufficiently large $n$ we have, for all $i$,
$$ \bar \alpha_i < 0 < c < \bar\beta_i, $$
using the definitions (\ref{eq:alpha-beta-bar}) of $\bar\alpha_i,\bar\beta_i$ along with the facts $\epsilon_1, \epsilon_2 = o(\log n)$ and $q < \omega < p$ (Lemma~\ref{lemma:omega-bounds}).
Our specific choice of $c$ is not crucial; we can in fact pick any deterministic $0 < c < \bar\beta_i$ provided that $c = \Theta(n \log n)$ and $\bar\beta_i - c = \Theta(n \log n)$ for all $i$.

Recall that our goal is to choose each $\gamma_v$ to lie in (or close to) the interval $[\alpha_v,\beta_v]$ subject to the condition $\sum_{v \in S_i} \gamma_v = c$ required by (\ref{eq:c}). To achieve this, define the deterministic quantity
\begin{equation}
\label{eq:kappa}
\kappa_i = \frac{c - \bar \alpha_i}{\bar \beta_i - \bar \alpha_i} \in (0,1).
\end{equation}
Note that we expect $c$ to lie roughly $\kappa_i$ fraction of the way through the interval $[\alpha_i,\beta_i]$ (which is the sum over $v \in S_i$ of the intervals $[\alpha_v,\beta_v]$). Mirroring this, we make a rough initial choice $\gamma_v'$ (for $v \in S_i$) that is $\kappa_i$ fraction of the way through the interval $[\alpha_v,\beta_v]$:
\begin{equation}
\label{eq:gamma-prime}
\gamma_v' = (1-\kappa_i) \alpha_v + \kappa_i \beta_v
\end{equation}
However, these do not satisfy $\sum_{v \in S_i} \gamma_v' = c$ on the nose -- rather, there is some error that is on the order of the difference between $\alpha_i$ and $\bar\alpha_i$. We thus introduce an additive correction term $\delta_i$ chosen to guarantee $\sum_{v \in S_i} \gamma_v = c$. For $v \in S_i$,
$$\gamma_v = \gamma_v' + \delta_i, \qquad \delta_i = \frac{1}{s_i} \left(c - \sum_{v \in S_i} \gamma_v' \right).$$

Recall that our goal was for $\gamma_v$ to lie within some $o(\log n)$ error from the interval $[\alpha_v,\beta_v]$. By construction we have $\gamma_v' \in [\alpha_v,\beta_v]$ and so we will have succeeded if we can show $\delta_i = o(\log n)$. This will be one of the goals of the next section.

\subsection{High-probability bounds for random variables}

In this section we establish bounds on various variables in the dual certificate that will hold with high probability $1 - o(1)$. We can treat the failure of these bounds as a low-probability failure event for the algorithm. 

First recall the following version of the Bernstein inequality:
\begin{lemma}[Bernstein inequality]
\label{lemma:bernstein}
If $X_1, \ldots, X_k$ are independent zero-mean random variables with $|X_i| \le 1$, then for any $t > 0$,
$$\Pr\left[\sum_i X_i \ge t \right] \le \exp\left(\frac{-\frac{1}{2} t^2}{\sum_i \Var[X_i] + \frac{1}{3} t}\right).$$
Note that by replacing $X_i$ with $-X_i$ we get the same bound for $\Pr\left[\, \sum_i X_i \le -t \, \right]$.
\end{lemma}

For each vertex $v$, let
\begin{equation}
\label{eq:Delta}
\Delta_v = \max_j |E(v,j) - \EE[E(v,j)]|
\end{equation}
where $j$ ranges over all communities, including that of $v$. Recall that for $v \notin S_j$, $E(v,j) \sim \mathrm{Binom}(s_j,q)$ and so $\EE[E(v,j)] = s_jq$; and for $v \in S_j$, $E(v,j) \sim \mathrm{Binom}(s_i-1,p)$ and so $\EE[E(v,j)] = (s_i-1)p$.

Our motivation for defining the quantity $\Delta_v$ is its appearance in the following bounds:
\begin{align*}
\left|\alpha_v - \frac{\bar\alpha_i}{s_i}\right| &= \left|p(s_i-1) - E(v,i)\right| \le \Delta_v \\
\left|\beta_v - \frac{\bar\beta_i}{s_i}\right| &= \left|\min_{j \ne i}\;[\omega s_j - E(v,j)] - (\omega-q) s_{\mathrm{min} \ne i}\right| \le \Delta_v \\
\numberthis\label{eq:gamma'bound}
\left|\gamma'_v - \frac{c}{s_i}\right| &= \left|(1-\kappa_i) \alpha_v + \kappa_i \beta_v - \frac{c}{s_i}\right| \\
&\le \left|(1-\kappa_i) \frac{\bar\alpha_v}{s_i} + \kappa_i \frac{\bar\beta_v}{s_i} - \frac{c}{s_i}\right| + \Delta_v \\
&\stackrel{(\ref{eq:kappa})}{=} \Delta_v
\end{align*}
where the third bound makes use of the first two.

Toward bounding $\Delta_v$, note that we can apply Bernstein's inequality (Lemma~\ref{lemma:bernstein}) to bound each $E(v,j) - \EE[E(v,j)]$, and take a union bound to obtain
$$ \Pr[ \Delta_v \geq t ] \leq 2 r \exp\left( \frac{- \frac12 t^2}{np + \frac{1}{3} t} \right). $$
Taking $t = \log n \log\log n$ and union bounding over all $v$, we see that, with high probability, $\Delta_v \leq \log n \log\log n$ for all $v$. But this will not quite suffice for the bounds we need. Instead, taking $t = \log n / (\log\log n)^2$, we see that $\Delta_v \leq \log n / (\log\log n)^2$ for most values of $v$, with a number of exceptions that, with high probability, does not exceed $n^{1 - 1 / (\log\log n)^5}$; for these exceptions, we fall back to the bound of $\log n \log\log n$ above. Above we have used the following consequence of Markov's inequality: if there are $n$ bad events, each occurring with probability $\le p$, then $\Pr[\text{at least } k \text{ bad events occur}] \le \frac{n p}{k}$.

For the sake of quickly abstracting away this two-tiered complication, we make the following three computations up front:
\begin{align*}
\numberthis\label{eq:Deltabound1}
\sum_v \Delta_v &\leq n^{1 - 1/(\log\log n)^5} \log n \log\log n + n \, \frac{\log n}{(\log\log n)^2} \\
&= \bigo( n \log n / (\log\log n)^2 ), \\
\numberthis\label{eq:Deltabound2}
\sum_v \Delta_v^2 &\leq n^{1 - 1/(\log\log n)^5} \log^2 n \, (\log\log n)^2 + n \, \frac{\log^2 n}{(\log\log n)^4} \\
&= \bigo( n \log^2 n / (\log\log n)^4 ), \\
\numberthis\label{eq:Deltabound3}
\sum_{u,v} (\Delta_u + \Delta_v)^2 & \leq 8n \cdot n^{1 - 1/(\log\log n)^5} \log^2 n \, (\log\log n)^2 + 4 n^2 \frac{\log^2 n}{(\log\log n)^4} \\
&= \bigo( n^2 \log^2 n / (\log\log n)^4 ),
\end{align*}
where the sums range over all vertices $u$,$v$. Note that in each case, the non-exceptional vertices dominate the bound. (One can easily compare two terms in the above calculations by computing the logarithm of each.)

Now we can show that $\delta_i$, the correction term from the previous section, is small. For any $i$ we have
\begin{align*}
\numberthis\label{eq:deltabound}
|\delta_i| &= \left|\frac{1}{s_i}(c - \sum_{v \in S_i} \gamma'_v)\right| \\
&\stackrel{(\ref{eq:gamma'bound})}{\le} \left|\frac{1}{s_i}(c - \sum_{v \in S_i} \frac{c}{s_i})\right| + \frac{1}{s_i}\sum_{v \in S_i} \Delta_v \\
&= \frac{1}{s_i} \sum_{v \in S_i} \Delta_v \\
&\stackrel{(\ref{eq:Deltabound1})}{=} \mathcal{O}(\log n / (\log\log n)^2).
\end{align*}

We will be interested in defining the quantity $\Delta'_v = \Delta_v + |\delta_i|$ (where $v \in S_i$) due to its appearance in the following bounds:
\begin{align*}
\numberthis\label{eq:gammabound}
\gamma_v &= \gamma'_v + \delta_i \stackrel{(\ref{eq:gamma'bound})}{=} \frac{c}{s_i} \pm \bigo(\Delta'_v), \\
\numberthis\label{eq:Rbound}
R_{vj} &= \omega s_j - E(v,j) - \gamma_v \\
&= (\omega-q) s_j - \frac{c}{s_i} \pm \bigo(\Delta'_v). \\
\end{align*}

Using the identity $(x+y)^2 \le 2(x^2+y^2)$ along with (\ref{eq:Deltabound2}), (\ref{eq:Deltabound3}) and (\ref{eq:deltabound}), we have with high probability
\begin{align*}
\numberthis\label{eq:Delta'bound2}
\sum_v (\Delta'_v)^2 &= \bigo( n \log^2 n / (\log\log n)^4 ), \\
\numberthis\label{eq:Delta'bound3}
\sum_{u,v} (\Delta'_u + \Delta'_v)^2 &= \bigo( n^2 \log^2 n / (\log\log n)^4 ).
\end{align*}

\subsection{Bounds on \texorpdfstring{$\nu_v$ and $R_{vj}$}{nu\_v and R\_vj}}

We can now prove two key results that we will need later: with high probability,
\begin{equation}
\label{eq:nubound}
\nu_v \ge \frac{\log n}{\log\log n} \quad \forall v
\end{equation}
and
\begin{equation}
\label{eq:Rg0}
R_{vj} > 0 \quad \forall j, \; \forall v \notin S_j.
\end{equation}
These results should not come as a surprise because they were more or less the motivation for defining the interval $[\alpha_v,\beta_v]$ for $\gamma_v$ in the first place. Since the $\nu_v$ values lie on the diagonal of $\Lambda$, the bound on $\nu_v$ is important for proving $\Lambda \succeq 0$ which we need for dual feasibility. Since $R_{vj}$ are the row sums of $\Gamma$, the bound on $R_{vj}$ is important for achieving $\Gamma_{S_i S_j} > 0$ which we need for Proposition~\ref{prop:dual-cert}. The specific quantity $\frac{\log n}{\log\log n}$ is not critical -- anything would suffice that is $o(\log n)$ yet large enough to dominate some error terms in a later calculation.

In the previous section we showed (\ref{eq:deltabound}) we showed $|\delta_i| = o(\log n)$ with high probability, and so we can choose the error terms $\epsilon_1,\epsilon_2$ from the definition of $\alpha_v,\beta_v$ (which, recall, are required to be $o(\log n)$) to absorb $\delta_i$. Specifically, let
\begin{align*}
\epsilon_1 &= \max_i |\delta_i| + \omega + \frac{\log n}{\log \log n} = o(\log n)\\
\epsilon_2 &= \max_i |\delta_i| + 1 = o(\log n).
\end{align*}

Recall that by Lemma~\ref{lemma:intervals-nonempty} the intervals $[\alpha_v,\beta_v]$ are all nonempty with high probability, and by construction we have $\gamma'_v \in [\alpha_v, \beta_v]$.

Now we will show (\ref{eq:nubound}): $\nu_v \ge \frac{\log n}{\log\log n}$. For $v \in S_i$ we have
\begin{align*}
\nu_v &= E(v,i) - \omega s_i + \gamma_v\\
&= E(v,i) - \omega s_i + \gamma_v' + \delta_i\\
&\ge E(v,i) - \omega s_i + \alpha_v + \delta_i\\
&\ge E(v,i) - \omega s_i + \omega (s_i-1) - E(v,i) + \epsilon_1 - \delta_i\\
&\ge \frac{\log n}{\log \log n},
\end{align*}
using the choice of $\epsilon_1$.

Now we show (\ref{eq:Rg0}): $R_{vj} > 0$. For $v \in S_i$ and $j \ne i$ we have
\begin{align*}
R_{vj} &= \omega s_j - E(v,j) - \gamma_v \\
&= \omega s_j - E(v,j) - \gamma_v' - \delta_i \\
&\ge \omega s_j - E(v,j) - \beta_v - \delta_i \\
&= \omega s_j - E(v,j) - \min_{k \neq i} \; [\omega s_k - E(v,k)] + \epsilon_2 - \delta_i \\
&\ge 1 > 0
\end{align*}
using the choice of $\epsilon_2$.

\subsection{Choice of \texorpdfstring{$\Gamma$}{Gamma}}

We have shown how to choose strictly positive row sums $R_{uj}$ and column sums $R_{vi}$ of $\Gamma_{S_i S_j}$ (for $i \neq j$). There is still considerable freedom in choosing the individual entries, but we will make what seems to be the simplest choice: we take $\Gamma_{S_i S_j}$ to be the unique \emph{rank-one} matrix satisfying these row and column sums, namely
$$(\Gamma_{S_i S_j})_{uv} = \frac{R_{uj} R_{vi}}{T_{ij}}$$
where $T_{ij}$ is the total sum of all entries of $\Gamma_{S_i S_j}$,
\begin{equation}
\label{eq:T}
T_{ij} = \sum_{u \in S_i} R_{uj} = \sum_{v \in S_j} R_{vi}.
\end{equation}
We showed earlier that this last equality is guaranteed by (\ref{eq:c}). As the row sums $R_{uj}$ are all positive with high probability (\ref{eq:Rg0}), it follows that $\Gamma_{S_i S_j} > 0$.

\subsection{PSD calculation for \texorpdfstring{$\Lambda$}{Lambda}}

We have already shown that if we choose $\nu_v$ and $R_{vj}$ according to (\ref{eq:nu}) and (\ref{eq:R}) respectively, then we have $\mathrm{span}\{\one_i - \one_j\} \subseteq \ker(\Lambda)$. In order to show $\Lambda \succeq 0$ with $\ker(\Lambda) = \mathrm{span}\{\one_i - \one_j\}$, we need to show
$$x^\top \Lambda x > 0 \quad \forall x \perp \mathrm{span}\{\one_i - \one_j\}.$$
The orthogonal complement of $\mathrm{span}\{\one_i - \one_j\}$ is spanned by block 0-sum vectors $Z = \{z \in \mathbb{R}^n : \sum_{v \in S_i} z_v = 0 \; \forall i\}$ plus the additional vector $y' = \sum_i \frac{1}{s_i} \one_i$. Let $y = \frac{y'}{\|y'\|} = \sum_i \frac{1}{s_i} \one_i / \sqrt{\sum_i \frac{1}{s_i}}$. Fix $x \perp \mathrm{span}\{\one_i - \one_j\}$ with $\|x\| = 1$ and write $x = \beta y + \sqrt{1-\beta^2} z$ for $\beta \in [0,1]$ and $z \in Z$ with $\|z\| = 1$. We have
\begin{equation}
\label{eq:xx}
x^\top \Lambda x = \beta^2 y^\top \Lambda y + 2 \beta \sqrt{1-\beta^2} z^\top \Lambda y + (1-\beta^2) z^\top \Lambda z.
\end{equation}
We will bound the three terms in (\ref{eq:xx}) separately. In particular, we will show that (with high probability):
\begin{align*}
\numberthis\label{eq:yzbounds}
y^\top &\Lambda y \phantom{|} = \Omega(\log n), \\
| z^\top &\Lambda y | = \bigo(\log n / (\log \log n)^2), \\
z^\top &\Lambda z \phantom{|} = \Omega(\log n / \log \log n ).
\end{align*}
Once we have this, we can (for sufficiently large $n$) rewrite (\ref{eq:xx}) as
$$ x^\top \Lambda x \geq \beta^2 \, C_1 \log n - 2 \beta \sqrt{1-\beta^2} \, C_2 \frac{\log n}{(\log\log n)^2} + (1-\beta^2) \, C_3 \frac{\log n}{\log \log n} $$
for some positive constants $C_1, C_2, C_3$. For sufficiently large $n$ we have
$$ (C_1 \log n)\left(C_3 \frac{\log n}{\log \log n}\right) > \left(C_2 \frac{\log n}{(\log\log n)^2}\right)^2, $$
which implies $x^\top \Lambda x > 0$ for all $\beta \in [0,1]$, completing the proof that $\Lambda \succeq 0$ with $\ker(\Lambda) = \mathrm{span}\{\one_i - \one_j\}$. It remains to show the three bounds in (\ref{eq:yzbounds}).

\subsubsection{Compute \texorpdfstring{$y^\top \Lambda y$}{yT Lambda y}}

\begin{align*}
y'^\top \Lambda y' &= y'^\top(\diag(\nu) + \omega J - A - \Gamma)y' \\
&= \sum_v \nu_v y_v^2 + \omega r^2 - \sum_{i,j} \frac{E(i,j)}{s_i s_j} - \sum_{i,j} \frac{T_{ij}}{s_i s_j} \\
& \stackrel{(a)}{=} \sum_i \frac{1}{s_i^2} \sum_{v \in S_i} \nu_v + \omega r^2 - \sum_i \frac{E(i,i)}{s_i^2} - \sum_{i \ne j} \frac{E(i,j)}{s_i s_j} - \sum_{i \ne j} (\omega - \frac{E(i,j)}{s_i s_j} - \frac{c}{s_i s_j}) \\
& \stackrel{(b)}{=} \sum_i \frac{1}{s_i^2} (E(i,i) - \omega s_i^2 + c) +\omega r^2 - \sum_i \frac{E(i,i)}{s_i^2} - \omega r(r-1) + c \sum_{i \ne j} \frac{1}{s_i s_j} \\
&= -\omega r + c \sum_i \frac{1}{s_i^2} + \omega r^2 - \omega r(r-1) + c \sum_{i \ne j} \frac{1}{s_i s_j} \\
&= c \left(\sum_i \frac{1}{s_i} \right)^2
\end{align*}

where (a) expands $\sum_{i,j} T_{ij}$ using (\ref{eq:T}),(\ref{eq:R}),(\ref{eq:c}), and (b) expands $\sum_{v \in S_i} \nu_v$ using (\ref{eq:nu}),(\ref{eq:c}). Therefore
\begin{equation}
\label{eq:yy}
y^\top \Lambda y = \frac{1}{\|y'\|^2} \, y'^\top \Lambda y' = c \sum_i \frac{1}{s_i}.
\end{equation}

Since in (\ref{eq:cchoice}) we chose $c$ to be $\Theta(n \log n)$, we have $y^\top \Lambda y = \Theta(\log n)$ as desired.

\subsubsection{Lower bound for \texorpdfstring{$z^\top \Lambda y$}{zT Lambda y}}
\label{sec:zLy}

For $v \in S_i$,
\begin{align*}
(\Lambda y')_v &= \frac{\nu_v}{s_i} + \omega r - \sum_j \frac{E(v,j)}{s_j} - \sum_{j \ne i} \frac{R_{vj}}{s_j} \\
&= \frac{1}{s_i}(E(v,i) - \omega s_i + \gamma_v) + \omega r - \sum_j \frac{E(v,j)}{s_j} - \sum_{j \ne i} \frac{1}{s_j} (\omega s_j - E(v,j) - \gamma_v) \\
&= \gamma_v \sum_i \frac{1}{s_i}
\end{align*}
and so
$$(\Lambda y)_v = \frac{1}{\|y'\|} (\Lambda y')_v = \gamma_v \sqrt{\sum_i \frac{1}{s_i}}.$$

Let $(\Lambda y)^Z$ denote the projection of the vector $\Lambda y$ onto the subspace $Z$. For $v \in S_i$ we have
$$(\Lambda y)^Z_v = (\gamma_v - \frac{1}{s_i} \sum_{v \in S_i} \gamma_v) \sqrt{\sum_i \frac{1}{s_i}} \stackrel{(\ref{eq:c})}{=} (\gamma_v - \frac{c}{s_i}) \sqrt{\sum_i \frac{1}{s_i}}.$$

Now we have
\begin{align*}
z^\top \Lambda y \ge - \|(\Lambda y)^Z\| &= -\sqrt{\sum_i \sum_{v \in S_i} (\gamma_v - \frac{c}{s_i})^2} \sqrt{\sum_i \frac{1}{s_i}} \\
&\stackrel{(\ref{eq:gammabound})}{\ge} -\sqrt{\sum_i \sum_{v \in S_i} \bigo(\Delta'_v)^2 } \sqrt{\sum_i \frac{1}{s_i}} \\
&\stackrel{(\ref{eq:Delta'bound2})}{=} -\bigo(\log n / (\log \log n)^2)
\end{align*}

\subsubsection{Lower bound for \texorpdfstring{$z^\top \Lambda z$}{zT Lambda z}}

Note that $J$, $\EE A + pI$ and $\EE \Gamma$ are block-constant and so the quadratic forms $z^\top J z$, $z^\top (\EE A + pI) z$ and $z^\top (\EE \Gamma)z$ are zero for $z \in Z$. Then
\begin{align*}
z^\top \Lambda z &= z^\top (\diag(\nu) + \omega J - A - \Gamma) z\\
&= z^\top\diag(\nu)z - z^\top (A - \EE A) z + pz^\top I z - z^\top (\Gamma - \EE \Gamma)z\\
& \ge \min_v \nu_v + p - \|A - \EE A\| - \|\Gamma - \EE \Gamma\|.
\end{align*}

Earlier we showed (\ref{eq:nubound})
$$ \min_v \nu_v \geq \log n / \log \log n $$
with high probability and we have $p = \Theta(\frac{\log n}{n}) = o(\log n/\log\log n)$. It remains to bound $\|A - \EE A\|$ and $\|\Gamma - \EE \Gamma\|$. The next two sections show that each of these two terms is $o(\log n / \log \log n)$ with high probability. It then follows that $z^\top \Lambda z = \Omega(\log n / \log\log n)$, as desired.

\subsubsection{Upper bound for \texorpdfstring{$\|A - \EE A\|$}{|A - EA|}}

Strong bounds for the spectral norm $\|A - \EE A\|$ have already appeared in the block model literature. Specifically, Theorem~5.2 of \cite{lr} is plenty stronger than we need; it follows immediately that
$$ \|A - \EE A \| \leq \bigo(\sqrt{\log n}) = o(\log n / \log \log n). $$

\subsubsection{Upper bound for \texorpdfstring{$\|\Gamma - \EE \Gamma\|$}{|Gamma - E Gamma|}}

Recall that $\Gamma_{S_i S_j}$ has row sums $R_{uj}$ and total sum (of all entries)
$$T_{ij} = \sum_{u \in S_i} R_{uj} \stackrel{(\ref{eq:R},\ref{eq:c})}{=} \omega s_i s_j - E(i,j) - c.$$

By applying Bernstein's inequality (Lemma~\ref{lemma:bernstein}) to $E(i,j)$ we get a high-probability bound for $T_{ij}$:
\begin{align*}
\numberthis\label{eq:Tbound}
T_{ij} &= \omega s_i s_j - E(i,j) - c \\
&= (\omega-q) s_i s_j - c \pm \bigo(\sqrt{n} \log n).
\end{align*}

We now compute $\EE \Gamma$. This is block-constant, by symmetry under permuting vertices within each community, and this constant must be
$$ \EE [\Gamma_{uv}] = \frac{1}{s_i s_j} \EE [T_{ij}] = \frac{(\omega-q) s_i s_j - c}{s_i s_j}, $$
where $u \in S_i$, $v \in S_j$, $i \neq j$.

We can now compute
\begin{align*}
\Gamma_{uv} - \EE \Gamma_{uv} &= \frac{R_{uj} R_{vi}}{T_{ij}} - \frac{(\omega - q) s_i s_j - c}{s_i s_j} \\
&= \frac{s_i s_j R_{uj} R_{vi} - ((\omega-q)s_i s_j - c) T_{ij}}{s_i s_j T_{ij}} \\
&\stackrel{(a)}= \frac{\pm \bigo(n^2 \log n (\Delta'_u + \Delta'_v) + n^2 \Delta'_u \Delta'_v)}{s_i s_j ((\omega - q) s_i s_j - c) + o(n^3 \log n)} \\
&\stackrel{(b)}{=} \frac{\pm \bigo(n^2 \log n (\Delta'_u + \Delta'_v) + n^2 \Delta'_u \Delta'_v)}{\Theta(n^3 \log n)} \\
&= \pm \bigo((\Delta'_u + \Delta'_v) / n + \Delta'_u \Delta'_v / (n \log n)),
\end{align*}
where in step (a), we appeal to the bounds (\ref{eq:Rbound}), (\ref{eq:Tbound}), causing cancellations in the high-order terms; and in (b) we have used the choice of $c$ (\ref{eq:cchoice}) to check that the denominator is $\Theta(n^3 \log n)$.

We will bound the spectral norm of $\Gamma - \EE \Gamma$ by its Frobenius norm. While this bound is often weak, it will suffice here as $\Gamma$ has constant rank, so that we only expect to lose a constant factor.
\begin{align*}
\| \Gamma - \EE \Gamma \| &\leq \| \Gamma - \EE \Gamma \|_F \\
&= \sqrt{\sum_{u,v} \left(\Gamma_{uv} - \EE \Gamma_{uv}\right)^2} \\
&= \sqrt{\sum_{u,v} \bigo\left(\frac{\Delta'_u + \Delta'_v}{n} + \frac{\Delta'_u \Delta'_v}{n \log n} \right)^2} \\
&\stackrel{(a)}{\le} \sqrt{\sum_{u,v} \bigo\left(\frac{\Delta'_u + \Delta'_v}{n}\right)^2} + \sqrt{\sum_{u,v} \bigo\left(\frac{\Delta'_u \Delta'_v}{n \log n} \right)^2} \\
&\le \frac{1}{n} \sqrt{\sum_{u,v} \bigo\left(\Delta'_u + \Delta'_v\right)^2} + \frac{1}{n \log n} \sqrt{\bigo\left(\sum_u (\Delta'_u)^2 \right)\left(\sum_v (\Delta'_v)^2 \right)} \\
&\stackrel{(b)}{=} \bigo(\log n / (\log\log n)^2) = o(\log n / \log \log n),
\end{align*}
where (a) uses the triangle inequality for the Euclidean norm and (b) uses the high-probability bounds (\ref{eq:Delta'bound2}),(\ref{eq:Delta'bound3}) for expressions involving $\Delta'_v$.

This completes the proof that $\Lambda \succeq 0$ with $\ker(\Lambda) = \mathrm{span}\{\one_i - \one_j\}$. We have now satisfied all conditions of Proposition~\ref{prop:dual-cert}, and so we may conclude Theorem~\ref{thm:recovery}: Programs~\ref{prog:known-primal} and \ref{prog:unknown-primal} achieve exact recovery with probability $1 - o(1)$.


\section{Acknowledgements}
The authors are indebted to Ankur Moitra for suggesting the problem, for providing guidance throughout the project, and for several enlightening discussions on semirandom models. We would also like to thank David Rolnick for a helpful discussion on MLEs for the block model, and Roxane Sayde for comments on a draft of this paper.

\bibliographystyle{alpha}
\bibliography{sbm}

\appendix

\section{Proof of Lemma~\ref{lemma:omega-bounds}}
\label{appendix:omega-bounds}

In this appendix we verify, for all $0 < q < p < 1$, that $q < \omega < p$, where
$$ \omega = \frac{\log(1-q)-\log(1-p)}{\log p - \log q + \log(1-q) - \log(1-p)}. $$
The proof is an elementary computation using the bound: $\frac{x-1}{x} \leq \log x \leq x-1$ for all $x > 0$, with both inequalities strict unless $x = 1$.

For the lower bound, we proceed as follows:
\begin{align*}
\frac{1}{\omega}
&= 1 + \frac{\log \frac{p}{q}}{\log \frac{1-q}{1-p}} \\
&< 1 + \frac{\frac{p}{q}-1}{(\frac{1-q}{1-p}-1) \big/ (\frac{1-q}{1-p})} = \frac{1}{q},
\end{align*}
and for the upper bound, we proceed similarly:
\begin{align*}
\frac{1}{\omega}
&= 1 + \frac{\log \frac{p}{q}}{\log \frac{1-q}{1-p}} \\
&> 1 + \frac{\frac{p/q - 1}{p/q}}{\frac{1-q}{1-p} - 1} = \frac{1}{p},
\end{align*}
so that the result follows by taking reciprocals.

\section{Proof of Proposition~\ref{prop:lower-bound-closed-form}}
\label{appendix:lower-bound-closed-form}

In this appendix, we establish a closed form for the CH-divergence in the planted partition model.

The CH-divergence is defined in \cite{as} as 
$$ D_+(i,j) = \sup_{t \in [0,1]} \sum_k \pi_k (t \tilde Q_{ik} + (1-t) \tilde Q_{jk} - \tilde Q_{ik}^t \tilde Q_{jk}^{1-t}), $$
summing over all communities $k$ including $i$ and $j$. (The limits on $t$ are unimportant: one can show that the supremum over $t \in \RR$ always lies in $[0,1]$.)

Note that if $\tilde Q_{ik} = \tilde Q_{jk}$ then the $k$ term of this sum vanishes. In the planted partition model, we have $\tilde Q_{ii} = \tilde Q_{jj} = \tilde p$, $\tilde Q_{ij} = \tilde Q_{ji} = \tilde q$, and $\tilde Q_{ik} = \tilde q = \tilde Q_{jk}$ for all other $k$. In particular, only the $i$ and $j$ terms of the sum will contribute. Thus:
\begin{align*}
D_+(i,j) &= \sup_t\; t \tilde p \pi_i + (1-t) \tilde q \pi_i + t \tilde q \pi_j + (1-t) \tilde p \pi_j - \pi_i \tilde p^t \tilde q^{1-t} - \pi_j \tilde p^{1-t} \tilde q^t \\
&= \sup_t\; t (\tilde p - \tilde q)(\pi_i-\pi_j) + \pi_i \tilde q (1 - (\tilde p / \tilde q)^t) + \pi_j \tilde p (1 - (\tilde p / \tilde q)^{-t}).
\end{align*}
Substituting $u = t/(\log \tilde p - \log \tilde q)$, we obtain
$$ D_+(i,j) = \sup_{u \geq 0}\; u \tau (\pi_i - \pi_j) + \pi_i \tilde q (1-e^u) + \pi_j \tilde p (1-e^{-u}) $$
where $\tau$ is as defined in Proposition~\ref{prop:lower-bound-closed-form}.

As the supremand is smooth and concave in $u$, we can set the derivative in $u$ equal to zero in order to maximize it:
$$ 0 = \tau (\pi_i - \pi_j) - \pi_i \tilde q e^u + \pi_j \tilde p e^{-u}, $$
which is quadratic in $e^u$, and can be solved via the quadratic formula:
$$ e^u = \frac{\tau (\pi_i-\pi_j) + \sqrt{\tau^2(\pi_i-\pi_j)^2 + 4 \pi_i \pi_j \tilde p \tilde q}}{2 \pi_i \tilde q}. $$

To obtain the simplest possible form for $D_+$, it is worth also solving for $e^{-u}$:
$$ e^{-u} = \frac{-\tau (\pi_i-\pi_j) + \sqrt{\tau^2(\pi_i-\pi_j)^2 + 4 \pi_i \pi_j \tilde p \tilde q}}{2 \pi_j \tilde p}. $$
Dividing these two expressions, we obtain
$$ e^{2u} = \frac{\pi_j \tilde p}{\pi_i \tilde q} \frac{\tau(\pi_i-\pi_j) + \gamma}{\tau(\pi_j-\pi_i) + \gamma}, $$
where $\gamma = \sqrt{\tau^2(\pi_i-\pi_j)^2 + 4 \pi_i \pi_j \tilde p \tilde q}$.
We now express $u$ as half the log of this quantity.

Substituting back into the divergence, we obtain
$$ D_+(i, j) = \pi_i \tilde q + \pi_j \tilde p - \gamma + \frac12 \tau (\pi_i - \pi_j) \log\left( \frac{\pi_j \tilde p}{\pi_i \tilde q} \cdot \frac{\tau(\pi_i-\pi_j)+\gamma}{\tau(\pi_j-\pi_i)+\gamma} \right), $$
thus proving the proposition.

\section{Proof of Proposition~\ref{prop:divergence-monotonicity}}
\label{appendix:divergence-monotonicity}

In Proposition~\ref{prop:lower-bound-closed-form}, we found that $D_+(i,j) = \eta(\tilde p, \tilde q, \pi_i,\pi_j)$ for a certain explicit function $\eta$. We wish to see that $\eta$ is monotone increasing in its third and fourth parameters. This implies, for example, that when checking whether exact recovery is possible in the planted partition model, it suffices to check that the divergence is at least $1$ between the two smallest communities.

Note that $\eta(a,b,\alpha c, \alpha d) = \alpha \eta(a,b,c,d)$, and that $\eta(a,b,c,d)=\eta(a,b,d,c)$, so it suffices to show that $\eta(\tilde p, \tilde q, s, 1)$ is monotone in $s$.

As $\eta$ is smooth, we will show that $\frac{\partial}{\partial s} \eta \geq 0$. We will show this, in turn, by showing that $\lim_{s \to 0} \frac{\partial}{\partial s} \eta \geq 0$ and that $\frac{\partial^2}{\partial s^2} \eta \geq 0$.

\begin{itemize}
\item
We first compute:
\begin{align*}
\frac{\partial}{\partial s} \eta(\tilde p,\tilde q,s,1) &= \frac{1}{2s} \left[ 2\tilde qs - \sqrt{\omega^2(s-1)^2 + 4\tilde p\tilde qs} \right. \\  &\left. + \tau \left(1-s+s\log\left( \frac{\tilde p}{\tilde qs} \frac{-\tau(s-1) + \sqrt{\tau^2(s-1)^2+4\tilde p\tilde qs} }{\tau(s-1) + \sqrt{\tau^2(s-1)^2+4\tilde p\tilde qs}} \right) \right) \right], \\
 \frac{\partial^2}{\partial s^2} \eta(\tilde p,\tilde q,s,1) &= \frac{\tau^2(1+s)^2 + 2\tilde p\tilde qs - \tau(1+s) \sqrt{\tau^2(s-1)^2 + 4\tilde p\tilde qs}}{2s^2 \sqrt{\tau^2(s-1)^2 + 4\tilde p\tilde qs}}.
\end{align*}

\item To see that the second partial is non-negative, it suffices to see that the numerator is non-negative:
$$ \tau^2(1+s)^2 + 2\tilde p\tilde qs \stackrel{?}{\geq} \tau(1+s)\sqrt{\tau^2(s-1)^2+4\tilde p\tilde qs} $$
Both sides are non-negative, so it is equivalent to compare their squares:
$$ \tau^4(1+s)^4 + 4\tilde p\tilde qs\tau^2(1+s)^2 + 4\tilde p^2\tilde q^2s^2 \stackrel{?}{\geq} \tau^2(1+s)^2(\tau^2(s-1)^2+4\tilde p\tilde qs) $$
which is evidently true when $s \geq 0$.

\item To see that the limit $\lim_{s \to 0} \frac{\partial}{\partial s} \eta$ is non-negative, we first compute it: 
$$ \lim_{s \to 0} \frac{\partial}{\partial s} \eta = \frac{-\tilde p\tilde q}{\tau}+\frac{\tilde p\log\frac{\tilde p}{\tau} - \tilde q\log\frac{\tilde q}{\tau}}{\log \tilde p - \log \tilde q}. $$
We now divide through through by $\tilde q$, and set $\alpha = \frac{\tilde p}{\tilde q}$ and $\beta = \frac{\tau}{\tilde q}$, noting that $1 \leq \beta \leq \alpha$:
\begin{align*}
\frac{1}{\tilde q} \lim_{s \to 0} \frac{\partial}{\partial s} \eta &= \alpha(1-\frac{1}{\beta}) + \beta \log \frac1\beta \\
&\geq \alpha(1-\frac1\beta) - \beta(1-\frac1\beta) \\
&= (\alpha-\beta)(1-\frac1\beta) \geq 0.
\end{align*}
\end{itemize}
This proves the proposition.

\end{document}